\documentclass[10pt]{article}
\usepackage{amsmath}
\usepackage{amssymb, amscd, amsthm}
\usepackage[all]{xy}
\usepackage[dvips]{graphicx}
\usepackage{verbatim}
\usepackage[perpage,symbol*]{footmisc}
\newtheorem{theorem}{Theorem}

\newtheorem{lemma}{\textbf{Lemma}}[section]
\newtheorem{remark}{\textbf{Remark}}[section]

\newtheorem{example}{\textbf{Example}}[section]

\usepackage[justification=centering]{caption}
\usepackage{longtable}
\newcommand{\tabincell}[2]{\begin{tabular}{@{}#1@{}}#2\end{tabular}}
\usepackage{multirow}
\usepackage{slashbox}

\newcommand{\F}{\mathbb{F}}

\begin{document}

\baselineskip 17pt
\title{\Large\bf New MDS Self-dual Codes over Finite Fields of Odd Characteristic}

\author{\large  Xiaolei Fang\quad \quad Khawla Lebed \quad\quad Hongwei Liu \quad\quad Jinquan Luo\footnote{The authors are with School of Mathematics
and Statistics \& Hubei Key Laboratory of Mathematical Sciences, Central China Normal University, Wuhan China.\newline
 E-mail: fangxiaolei@mails.ccnu.edu.cn(X.Fang), labad.k@yahoo.com(K.Lebed),
hwliu@mail.ccnu.edu.cn(H.Liu), luojinquan@mail.ccnu.edu.cn(J.Luo)}}
\date{}
\maketitle

{\bf Abstract}:
In this paper, we produce new classes of MDS self-dual codes via (extended) generalized Reed-Solomon codes over finite fields of odd characteristic.
Among our constructions, there are many MDS self-dual codes with new parameters which have never been reported. For odd prime power $q$ with $q$ square,
the total number of lengths for MDS self-dual codes over $\mathbb{F}_q$ presented in this paper is much more than those in all the previous results.

{\bf Key words}: MDS code, Self-dual code, Generalized Reed-Solomon code, Extended generalized Reed-Solomon code
\section{Introduction}

 \quad\; Let $\mathbb{F}_{q}$ be the finite field with $q$ elements, where $q$ is a prime power. A linear code $C$ of length $n$, dimension $k$ and
minimum distance $d$ over $\mathbb{F}_{q}$ is usually called a $q$-ary $[n,k,d]$ code. If the parameters of the code $C$ attach the Singleton bound:
$k+d=n+1$, then $C$ is called a maximum distance separable (MDS) code. MDS codes are widely applied in various occasions due to their nice properties, see
[\ref{BR}, \ref{KKO}, \ref{SR}].

The dual code of a linear code $C$ in $\mathbb{F}_{q}^{n}$, denoted by $C^{\perp}$, is a linear subspace of $\mathbb{F}_{q}^{n}$,
which is orthogonal to $C$. If $C=C^{\perp}$, $C$ is called a self-dual code. Self-dual codes have important applications in coding theory [\ref{Rain}],
cryptography [\ref{CDG}, \ref{DMS}, \ref{MAS}], combinatorics [\ref{BHM}, \ref{MS}] and other related areas.

MDS self-dual codes have good properties due to their optimality with respect to the Singleton bound and their self-duality, which have attracted
a lot of attention in recent years. There are various ways to construct MDS self-dual codes. They mainly are: (1). orthogonal designs, see
[\ref{GK2}, \ref{HK1}, \ref{HK2}]; (2).  building up technique, see [\ref{KL1}, \ref{KL2}]; (3). constacyclic codes, see
[\ref{KZT}, \ref{TW}, \ref{YC}]; (4). (generalized and/or extended) Reed-Solomon codes, see
[\ref{FF}. \ref{GKL}, \ref{JX}, \ref{LLL}, \ref{TW}, \ref{Yan}, \ref{ZF}].

Parameters of MDS self-dual codes are completely characterized by their lengths $n$, that is, $\left[n,\frac{n}{2},\frac{n}{2}+1\right]$.
The problem for constructing different MDS self-dual codes can be transformed to find MDS self-dual codes with
different lengths. In [\ref{GG}], Grassl and Gulliver showed that the problem has been completely solved over the finite fields of characteristic $2$.
But the constructions of MDS self-dual codes on finite fields of odd characteristic are still far from complete. For example, if $q=83^{2}$, more than 3000 MDS self-dual codes with different even lengths possibly exist assuming MDS conjecture  is valid (MDS conjecture says that the length of nontrivial $q$-ary MDS code with $q$ odd prime power, is bounded
by $q+1$). But up to now, only 702 $q$-ary MDS self-dual codes of different even lengths have been constructed. In [\ref{JX}], Jin and Xing
constructed some classes of new MDS self-dual codes through generalized Reed-Solomon codes. In [\ref{Yan}], Yan generalized the technique in [\ref{JX}]
and constructed several classes of MDS self-dual codes via generalized Reed-Solomon codes and extended generalized Reed-Solomon codes. In [\ref{LLL}],
Lebed, Liu and Luo produced more classes of MDS self-dual codes based on [\ref{JX}] and [\ref{Yan}]. All the known results on the systematic constructions of MDS self-dual codes are depicted in Table 1.

\begin{center}
\begin{longtable}{|c|c|c|}  
\caption{Known systematic construction on MDS self-dual codes of length $n$ \\ ($\eta$ is the quadratic character of $\mathbb{F}_{q}$) } \\ \hline
$q$ & $n$ even & Reference\\  \hline
$q$ even  &  $n \leq q$   & [\ref{GG}] \\ \hline
$q$ odd & $n=q+1$ & [\ref{GG}]\\ \hline $q$ odd & $(n-1)|(q -1)$, $\eta(1 - n) = 1$ &   [\ref{Yan}] \\ \hline
$q$ odd & $(n-2)|(q - 1)$, $\eta(2 - n) = 1$ &   [\ref{Yan}]\\ \hline
$q = r^{s}\equiv 3\pmod{4}$ &  $n-1= p^m \mid(q-1)$, prime $p\equiv 3\pmod{4}$ and  $m$ odd &  [\ref{GUE}]\\ \hline
$q = r^{s}$, $r \equiv 1 \pmod{4}$, $s$ odd &  $n-1= p^m\mid (q-1)$, $m$ odd and prime $p \equiv 1\pmod{4}$  &  [\ref{GUE}]\\ \hline
$q = r^{s}$ , $r$ odd, $s\geq 2$ & $n = lr$,  $l$ even and $2l|(r - 1)$ &   [\ref{Yan}] \\ \hline

$q = r^{s}$ , $r$ odd, $s \geq 2$ & $n = lr$,  $l$ even , $(l - 1)|(r - 1)$ and $\eta(1 - l)=1$ &   [\ref{Yan}] \\ \hline

$q = r^{s}$ , $r$ odd, $s \geq 2$ & $n = lr + 1$, $l$ odd , $l|(r - 1)$ and $\eta(l) = 1$  &   [\ref{Yan}] \\ \hline
 $q = r^{s}$ , $r$ odd, $s \geq 2$ & $n = lr + 1$, $l$ odd , $(l - 1)|(r - 1)$ and $\eta(l - 1) = \eta(-1) = 1$ &  [\ref{Yan}] \\ \hline

$q=r^2$  & $n \leq r$  & [\ref{JX}] \\ \hline
$q = r^2, r \equiv 3\pmod{4}$  &  $n= 2tr$ for any $t \leq \frac{r - 1}{2}$ &   [\ref{JX}]\\ \hline

$q = r^2$, $r$ odd & $n = tr$, $t$ even and $1 \leq t \leq r$ &   [\ref{Yan}] \\ \hline

 $q = r^2$, $r$ odd & $n = tr + 1$,  $t$ odd and $1 \leq t \leq r$ &   [\ref{Yan}] \\ \hline

$q  \equiv 1\pmod{4}$ &  $ n|(q - 1), n < q - 1$ &   [\ref{Yan}] \\ \hline
$q \equiv 1\pmod{4}$ &  $4^{n}\cdot n^{2} \leq q$ &  [\ref{JX}]\\ \hline

  $q = p^k $, odd prime $p$ & $n= p^r + 1$, $r|k$ &   [\ref{Yan}] \\ \hline
$q = p^k $, odd prime $p$ & $n= 2p^e$, $1 \leq e < k$, $\eta(-1) = 1$&  [\ref{Yan}] \\ \hline
$q=r^2$, $r$ odd & $n=tm$, $1\leq t \leq \frac{r-1}{\gcd(r-1,m)}$, $\frac{q-1}{m}$ even &  [\ref{LLL}] \\ \hline
$q=r^2$, $r$ odd & $n=tm+1$, $tm$ odd, $1\leq t \leq \frac{r-1}{\gcd(r-1,m)}$ and $m|(q-1)$  & [\ref{LLL}]\\ \hline
$q=r^2$, $r$ odd & $n=tm+2$, $tm$ even, $1\leq t \leq \frac{r-1}{\gcd(r-1,m)}$ and $m|(q-1)$   &   [\ref{LLL}]\\\hline
$q=p^m$, odd prime $p$ & $n= 2tp^e$, $2t \mid (p-1)$ and $e<m$, $\frac{q-1}{2t}$ even &  [\ref{LLL}]\\ \hline
$q=p^{m}$, $m$ even, odd prime $p$ & $n=2tr^l$ with $r=p^s$, $s\mid\frac{m}{2}$, $0\leq l\leq \frac{m}{s}$ and $1\leq t\leq\frac{r-1}{2}$ & [\ref{FF}]\\ \hline
$q=p^{m}$, $m$ even, odd prime $p$ &\tabincell{c}{$n=(2t+1)r^l+1$ with $r=p^s$, $s\mid\frac{m}{2}$, $0\leq l<\frac{m}{s}$ \\and $0\leq t\leq\frac{r-1}{2}$ or $l=\frac{m}{s}$, $t=0$} & [\ref{FF}]\\ \hline
$q=p^m\equiv1\,(\mathrm{mod}\,4)$ & $n= p^l+1$ with $0\leq l\leq m$ &   [\ref{FF}] \\ \hline
\end{longtable}
 \end{center}

Based on [\ref{JX}], [\ref{LLL}] and [\ref{Yan}], we give more constructions of MDS self-dual codes in this paper. Among our constructions, there are
several MDS self-dual codes with new parameters (see Table 2). In particular, for square $q$, we can produce much more MDS self-dual codes than the
previous works.

This paper is organized as follows. In Section 2, we will introduce some basic knowledge and useful results on (extended) generalized Reed-Solomon codes.
In Section 3, we will present our main results on the constructions of MDS self-dual codes. In Section 4, we will make a conclusion.


\begin{center}
\begin{longtable}{|c|c|c|}  
\caption{Our  results} \\\hline
$q$ &   $n$ even  &  Reference \\\hline
$q=r^2$, $r$ odd & $n=tm$, $1\leq t \leq \frac{r+1}{\gcd(r+1,m)}$, $\frac{q-1}{m}$ even & Theorem \ref{thmA1} (i) \\ \hline

$q=r^2$, $r$ odd  &\tabincell{c}{$n=tm+2$, $tm$ even(except when $t$ is even, $m$ is even\\
 and $r\equiv1\,(\mathrm{mod}\,4)$), $1\leq t \leq \frac{r+1}{\gcd(r+1,m)}$ and $m|(q-1)$}   &  Theorem \ref{thmA1} (ii) \\\hline
 $q=r^2$, $r$ odd & $n=tm+1$, $tm$ odd, $2\leq t \leq \frac{r+1}{2\gcd(r+1,m)}$ and $m|(q-1)$  & Theorem \ref{thmA2}\\ \hline
   $q=r^2$, $r$ odd & \tabincell{c}{$n=tm$, $1\leq t \leq \frac{s(r-1)}{\gcd(s(r-1),m)}$, $s$ even, $s|m$,\\ $\frac{r+1}{s}$ even and $\frac{q-1}{m}$ even}  & Theorem \ref{thmC1} (i)\\ \hline
      $q=r^2$, $r$ odd & \tabincell{c}{$n=tm+2$, $1\leq t \leq \frac{s(r-1)}{\gcd(s(r-1),m)}$, $s$ even, $s|m$,\\ $s\mid r+1$ and $m|(q-1)$}  & Theorem \ref{thmC1} (ii)\\ \hline
$q=p^{2s}$, odd prime $p$ & $n= p^{2e}+1$, $1\leq e\leq s$ & Theorem \ref{thmD}\\ \hline
$q=p^{km}$, odd prime $p$ & $n= 2tp^{ke}$, $2t | (p^k-1)$ and $e\leq m-1$,  $\frac{q-1}{2t}$ even & Theorem \ref{thmE}\\ \hline
\end{longtable}
 \end{center}

\section{Preliminaries}
 \quad\; In this section, we introduce some basic notation and useful results on (extended) generalized Reed-Solomon codes
(or (extended) $\mathbf{GRS}$ codes for short). Readers are referred to [\ref{MS}, Chapter 10] for more details.

Let $\mathbb{F}_{q}$ be the finite field with $q$ elements and $n$ be an integer with $1\leq n\leq q$. Choose two $n$-tuples
$\overrightarrow{v}=(v_{1},v_{2},\ldots,v_{n})$ and $\overrightarrow{a}=(\alpha_{1},\alpha_{2},\ldots,\alpha_{n})$, where $v_{i}\in\mathbb{F}_{q}^{*}$,
$1\leq i\leq n$ ($v_{i}$ may not be distinct) and $\alpha_{i}$, $1\leq i\leq n$ are distinct elements in $\mathbb{F}_{q}$. For an integer $k$ with
$1\leq k\leq n$, the $\mathbf{GRS}$ code of length $n$ associated with $\overrightarrow{v}$ and $\overrightarrow{a}$ is defined as follows:
\begin{equation}\label{def GRS}
\mathbf{GRS}_{k}(\overrightarrow{a},\overrightarrow{v})=\{(v_{1}f(\alpha_{1}),\ldots,v_{n}f(\alpha_{n})):f(x)\in\mathbb{F}_{q}[x],\mathrm{deg}(f(x))\leq k-1\}.
\end{equation}
The code $\mathbf{GRS}_{k}(\overrightarrow{a},\overrightarrow{v})$ is a $q$-ary $[n,k]$ MDS code and its dual is also MDS [\ref{MS}, Chapter 11].

We define
\begin{equation*}\label{def GRS}
L_{\overrightarrow{a}}(\alpha_{i})=\prod_{1\leq j\leq n,j\neq i}(\alpha_{i}-\alpha_{j}).
\end{equation*}
Let $\Box_{q}$ denote the set of nonzero squares of $\mathbb{F}_{q}$. The following result is useful in our constructions and it has been shown in [\ref{JX}].

\begin{lemma}([\ref{JX}], Corollary 2.4)\label{y1}
Let $n$ be an even integer and $k=\frac{n}{2}$. If there exists $\lambda\in\mathbb{F}_{q}^{*}$ such that
$\lambda L_{\overrightarrow{a}}(\alpha_{i})\in\Box_{q}$ for all $1\leq i\leq n$, then there exists $\overrightarrow{v}=(v_{1},\ldots,v_{n})$ with
$v_{i}^{2}=\frac{1}{\lambda L_{\overrightarrow{a}}(\alpha_{i})}$ such that the code $\mathbf{GRS}_{k}(\overrightarrow{a},\overrightarrow{v})$ defined
in (1) is an MDS self-dual code of length $n$.
\end{lemma}

Moreover, extended $\mathbf{GRS}$ code can also be applied to the construction of MDS self-dual codes. For $\overrightarrow{v}=(v_1,\ldots,v_{n-1})$ and
$\overrightarrow{a}=(a_1,\ldots,a_{n-1})$, the extended $\mathbf{GRS}$ code of length $n$ associated with $\overrightarrow{v}$ and $\overrightarrow{a}$
is defined as follows:

\begin{equation}\label{def extended GRS}
\mathbf{GRS}_{k}(\overrightarrow{a},\overrightarrow{v},\infty)=\{(v_{1}f(\alpha_{1}),\ldots,v_{n-1}f(\alpha_{n-1}),f_{k-1}):f(x)\in\mathbb{F}_{q}[x], \mathrm{deg}(f(x))\leq k-1\},
\end{equation}
where $f_{k-1}$ is the coefficient of $x^{k-1}$ in $f(x)$. The code $\mathbf{GRS}_{k}(\overrightarrow{a},\overrightarrow{v},\infty)$ is a
$q$-ary $[n,k]$ MDS code and its dual is also MDS [\ref{MS}, Chapter 11].

We present another two useful results, which have been shown in [\ref{Yan}].

\begin{lemma}([\ref{Yan}], Lemma 2)\label{y2}
Let $n$ be an even integer and $k=\frac{n}{2}$. If $-L_{\overrightarrow{a}}(\alpha_{i})\in\Box_{q}$ for all $1\leq i\leq n-1$,
then there exists $\overrightarrow{v}=(v_{1},\ldots,v_{n})$ with $v_{i}^{2}=-\frac{1}{ L_{\overrightarrow{a}}(\alpha_{i})}$ such that the code
$\mathbf{GRS}_{k}(\overrightarrow{a},\overrightarrow{v},\infty)$ defined in (2) is an MDS self-dual code of length $n$.
\end{lemma}

\begin{lemma}([\ref{Yan}], Lemma 3)\label{y3}
Let $m\mid q-1$ be a positive integer and let $\alpha\in\mathbb{F}_{q}$ be a primitive $m$-th root of unity.
Then for any $1\leq i\leq m$, $$\prod_{1\leq j\leq m, j\neq i}\left(\alpha^{i}-\alpha^{j}\right)=m\alpha^{-i}.$$
\end{lemma}

\section{Main Results}

\quad\; In this section, we will give several new constructions of MDS self-dual codes utilizing the multiplicative group structure of
$\mathbb{F}_{q}^{*}$ and the additive group structure on $\mathbb{F}_{q}$.

Let $g$ be a generator(primitive element) of $\mathbb{F}_{q}^{*}$. Let $S$ and $T$ be two cyclic subgroups of $\F_q^*$.
By the second fundamental theorem of group homomorphism,
\begin{equation}\label{ST}
S\Big{/}\left(S\cap T\right)\simeq \left(S\times T\right)\Big{/}T\leq \mathbb{F}_q^*\Big{/} T.
\end{equation}

\begin{theorem}\label{thmA1}
Let $q=r^{2}$, where $r$ is an odd prime power. Suppose $m\mid q-1$. For $1\leq t\leq \frac{r+1}{\gcd(r+1,m)}$, and $tm$ even,

(i). if $\frac{q-1}{m}$ is even and $n=tm$, then there exists a $q$-ary $[n,\frac{n}{2}]$ MDS self-dual code.

(ii). if $n=tm+2$, then there exists a $q$-ary $[n,\frac{n}{2}]$ MDS self-dual code except the case that $t$ is even, $m$ is even and $r\equiv1\,(\mathrm{mod}\,4)$.
\end{theorem}

\begin{proof}
(i). Put $S=\langle g^{r-1}\rangle$ and $T=\langle g^{\frac{q-1}{m}}\rangle$ in (\ref{ST}). Let $B=\{g^{(r-1)i_{1}},\ldots,g^{(r-1)i_{t}}\}$ be a set of
coset representatives of $(S\times T)/T$ with $0\leq i_{1}<\cdots<i_{t}<r+1$. Denote by $I=\{{i_{1},\ldots,i_{t}}\}$, $A=i_{1}+\cdots+i_{t}$ and
$$\overrightarrow{a}=\left(g^{\frac{k}{m}\cdot(q-1)+(r-1)z}\right)_{0\leq k\leq m-1, z\in I}.$$
Obviously, the entries of $\overrightarrow{a}$ are distinct in $\mathbb{F}_{q}^*$. We will show that there exists
$\overrightarrow{v}\in\left(\mathbb{F}_{q}^{*}\right)^{n}$ such that $\mathbf{GRS}_{\frac{n}{2}}(\overrightarrow{a},\overrightarrow{v})$
is an MDS self-dual code of length $n=tm$.

For any $z \in I$ and $0\leq k\leq m-1$,
\begin{equation}\label{La}
\begin{aligned}
L_{\overrightarrow{a}}(g^{\frac{k}{m}\cdot(q-1)+(r-1)z})=
&{\prod\limits_{0 \leq j \leq m-1, j\neq k}{(g^{\frac{k}{m}\cdot(q-1)+(r-1)z}-g^{\frac{j}{m}\cdot(q-1)+(r-1)z})}}&\\
&\cdot\prod\limits_{l\in I, l\neq z}{\prod\limits_{j=0}^{m-1}{(g^{\frac{k}{m}\cdot(q-1)+(r-1)z}-g^{\frac{j}{m}\cdot(q-1)+(r-1)l})}}.&\\
\end{aligned}
\end{equation}
By Lemma \ref{y3},
\begin{equation}\label{La1}
{\prod\limits_{0 \leq j \leq m-1, j\neq k}{(g^{\frac{k}{m}\cdot(q-1)+(r-1)z}-g^{\frac{j}{m}\cdot(q-1)+(r-1)z})}}=
m\cdot \left(g^{\frac{k}{m}\cdot(q-1)+(r-1)z}\right)^{m-1}.
\end{equation}
Note that $x^m-y^m=\prod\limits_{j=0}^{m-1}{(x-g^{\frac{j}{m}\cdot(q-1)} y)}$. Set $x=g^{\frac{k}{m}\cdot(q-1)+(r-1)z}$ and $y=g^{(r-1)l}$.
It follows that
\begin{equation}\label{La2}
\prod\limits_{j=0}^{m-1}{(g^{\frac{k}{m}\cdot(q-1)+(r-1)z}-g^{\frac{j}{m}\cdot(q-1)+(r-1)l})}=g^{(r-1)zm}-g^{(r-1)lm}.
\end{equation}
By substituting (\ref{La1}) and (\ref{La2}) into (\ref{La}),
\begin{equation*}
\begin{aligned}
L_{\overrightarrow{a}}(g^{\frac{k}{m}\cdot(q-1)+(r-1)z})= m\cdot \left(g^{\frac{k}{m}\cdot(q-1)+(r-1)z}\right)^{m-1}\cdot \prod\limits_{l\in I, l\neq z}{(g^{(r-1)zm}-g^{(r-1)lm})}.
\end{aligned}
\end{equation*}
Let $u=\prod\limits_{l\in I,l\neq z}(g^{(r-1)zm}-g^{(r-1)lm})$. We calculate
\[
\begin{array}{rcl}

u^{r}
 &=&
 {\prod\limits_{l\in I, l\neq z}{(g^{-(r-1)zm}-g^{-(r-1)lm})}}
 = \prod\limits_{l\in I, l\neq z}{g^{-(r-1)(l+z)m}(g^{(r-1)lm}-g^{(r-1)zm})}\\[2mm]
 &=&(-1)^{t-1}\cdot g^{-\left(\sum\limits_{l\in I,l\neq z}l+(t-1)z\right)(r-1)m}\cdot u
 =(-1)^{t-1}\cdot g^{-\left(A+(t-2)z\right)(r-1)m}\cdot u.\\[2mm]
\end{array}
 \]
So $u^{r-1}=(-1)^{t-1}\cdot g^{-\left(A+(t-2)z\right)(r-1)m}$. Since $-1=g^{\frac{r^{2}-1}{2}}$, then
$$u^{r-1}=g^{\frac{r^{2}-1}{2}\cdot(t-1)}\cdot g^{-(r-1)\cdot(A+(t-2)z)m}.$$ It follows that
$$u=g^{\frac{r+1}{2}\cdot(t-1)-(A+(t-2)z)m+i(r+1)}\,\,\mathrm{for\,\,some}\,\,i .$$

It is obvious that $ m\cdot \left(g^{\frac{k}{m}\cdot(q-1)+(r-1)z}\right)^{m-1}\in\Box_{q}$. We take
$\lambda=g^{\frac{r+1}{2}\cdot(t-1)-mA}\in\mathbb{F}_{q}^{*}$. Since $tm$ is even, we obtain that
$\lambda L_{\overrightarrow{a}}(g^{\frac{k}{m}\cdot(q-1)+(r-1)z})\in\Box_{q}$.
Choose $v_{z,k}^{2}=\left(\lambda L_{\overrightarrow{a}}(g^{\frac{k}{m}\cdot(q-1)+(r-1)z})\right)^{-1}$ with $v_{z,k}\in\mathbb{F}_{q}^{*}$.
Define $$\overrightarrow{v}=(v_{i_{1},0},\ldots,v_{i_{1},m-1},\ldots, v_{i_{t},0},\ldots,v_{i_{t},m-1}).$$
By Lemma \ref{y1}, $\mathbf{GRS}_{\frac{n}{2}}(\overrightarrow{a},\overrightarrow{v})$ is an MDS self-dual code.
Therefore, there exists a $q$-ary $[n,\frac{n}{2}]$ MDS self-dual code with length $n=tm$.\\

(ii). As in (i), we let
$$\overrightarrow{a}=\left(0,g^{(r-1)i_{1}}, g^{\frac{1}{m}\cdot(q-1)+(r-1)i_{1}},\ldots,g^{\frac{m-1}{m}\cdot(q-1)+(r-1)i_{1}},\ldots,
g^{(r-1)i_{t}},\ldots,g^{\frac{m-1}{m}\cdot(q-1)+(r-1)i_{t}}\right).$$
We will find $\overrightarrow{v}\in\left(\mathbb{F}_{q}^{*}\right)^{n}$ such that
$\mathbf{GRS}_{\frac{n}{2}}(\overrightarrow{a},\overrightarrow{v},\infty)$ is an MDS self-dual code of length $n=tm+2$.

For any $0\leq j\leq m-1$ and for any $l\in I$, $I=\{i_{1},\cdots,i_{t}\}$,
\begin{equation*}
\begin{aligned}
L_{\overrightarrow{a}}(g^{\frac{k}{m}\cdot(q-1)+(r-1)z})=
&g^{\frac{k}{m}\cdot(q-1)+(r-1)z}\cdot{\prod\limits_{0 \leq j \leq m-1, j\neq k}{(g^{\frac{k}{m}\cdot(q-1)+(r-1)z}-g^{\frac{j}{m}\cdot(q-1)+(r-1)z})}}\cdot&\\
&\prod\limits_{l\in I, l\neq z}{\prod\limits_{j=0}^{m-1}{(g^{\frac{k}{m}\cdot(q-1)+(r-1)z}-g^{\frac{j}{m}\cdot(q-1)+(r-1)l})}}&\\
=& m\cdot g^{(r-1)zm}\cdot \prod\limits_{l\in I, l\neq z}{(g^{(r-1)zm}-g^{(r-1)lm})}&\\
\end{aligned}
\end{equation*}
 and
 \[
 \begin{array}{rcl}
 L_{\overrightarrow{a}}(0)
 &=&\prod\limits_{l\in I}\prod\limits_{j=0}^{m-1}\left(0-g^{\frac{j}{m}\cdot(q-1)+(r-1)l}\right)
=(-1)^{(m+1)t}\cdot\prod\limits_{l\in I}g^{(r-1)lm}=\pm\prod\limits_{l\in I}g^{(r-1)lm}.\\[2mm]
\end{array}
 \]

Denote $u=\prod\limits_{l\in I,l\neq z}(\beta^{zm}-\beta^{lm})$. We obtain $u=g^{\frac{r+1}{2}\cdot(t-1)-(A+(t-2)z)m+i(r+1)}$ for some $i$,
in the same way as (i). The following cases are considered.

$\mathbf{Case \,\,1}$: If $t$ is odd and $m$ is even, we have $\frac{r+1}{2}\cdot(t-1)-(A+(t-2)z)m$ is even. It follows that $u\in\Box_{q}$.

$\mathbf{Case \,\,2}$: If $t$ is even and $r\equiv3\,(\mathrm{mod}\,4)$, we can choose $i_{1},\ldots,i_{t}$ such that $A=i_1+\cdots+i_t$ is even.
It follows that $\frac{r+1}{2}\cdot(t-1)-(A+(t-2)z)m$ is even. Hence $u\in\Box_{q}$.

$\mathbf{Case \,\,3}$: If $t$ is even, $m$ is odd and $r\equiv1\,(\mathrm{mod}\,4)$, we can choose $i_{1},\ldots,i_{t}$ such that $A$ is an odd integer.
It follows that $\frac{r+1}{2}\cdot(t-1)-(A+(t-2)z)m$ is even. Hence $u\in\Box_{q}$.

Note that $g^{r-1},m,-1\in\Box_{q}$. As a result, one always has
$L_{\overrightarrow{a}}(g^{\frac{k}{m}\cdot(q-1)+(r-1)z}),L_{\overrightarrow{a}}(0)\in\Box_{q}$.
It is easy to verify that $-L_{\overrightarrow{a}}(g^{\frac{k}{m}\cdot(q-1)+(r-1)z}),-L_{\overrightarrow{a}}(0)\in\Box_{q}$. We choose
$v_{z,k}^{2}=-\left(L_{\overrightarrow{a}}(g^{\frac{k}{m}\cdot(q-1)+(r-1)z})\right)^{-1}$ and $v_{0}^{2}=-\left(L_{\overrightarrow{a}}(0)\right)^{-1}$,
with $v_{z,k},v_{0}\in\mathbb{F}_{q}^{*}$. Define
$$\overrightarrow{v}=(v_{0},v_{i_{1},0},\ldots,v_{i_{1},m-1},\ldots, v_{i_{t},0},\ldots,v_{i_{t},m-1}).$$
By Lemma \ref{y2}, $\mathbf{GRS}_{\frac{n}{2}}(\overrightarrow{a},\overrightarrow{v},\infty)$ is an MDS self-dual code with length $n=tm+2$,
except the case that $t$ is even, $m$ is even and $r\equiv1\,(\mathrm{mod}\,4)$.
\end{proof}

\begin{example}
Let $r=151$, $q=151^{2}$, $m=6$ and $t=71$. Then $\frac{r+1}{\gcd(r+1,m)}=\frac{152}{2}=76>71=t$. By Theorem \ref{thmA1}, there exists an MDS self-dual
code of length $n=tm=426$. This is a new parameter of MDS self-dual code.
\end{example}

\begin{theorem}\label{thmA2}
Let $q=r^2$, where $r$ is an odd prime power. Suppose $m|(q-1)$. If $1\leq t\leq\frac{r+1}{2\gcd(r+1,m)}$, $tm$ is odd and $n=tm+1$,
then there exists an $[n, \frac{n}{2}]$ MDS self-dual code over $\mathbb{F}_{q}$.
\end{theorem}
\begin{proof}
Put $S=\langle g^{r-1}\rangle$ and $T=\langle g^{\frac{q-1}{m}}\rangle$ in (\ref{ST}). Let $B=\{g^{(r-1)i_{1}},\ldots,g^{(r-1)i_{t}}\}$ be a set of
coset representatives of $(S\times T)/T$ with $0\leq i_{1}<\cdots<i_{t}<r+1$ and
$i_j$ ($1\leq j\leq t$) even. Denote by $I=\{{i_{1},\ldots,i_{t}}\}$, $A=i_{1}+\cdots+i_{t}$ and
$$\overrightarrow{a}=\left(g^{\frac{k}{m}\cdot(q-1)+(r-1)z}\right)_{0\leq k\leq m-1, z\in I}.$$
Obviously, the entries of $\overrightarrow{a}$ are distinct in $\mathbb{F}_{q}^*$.
The main goal is to find $\overrightarrow{v}$ such that $\mathbf{GRS}_{\frac{n}{2}}(\overrightarrow{a},\overrightarrow{v},\infty)$ is an MDS self-dual
code. Similarly as in Theorem \ref{thmA1} (i), for $z=i_j$, $1\leq j\leq t$ and $0\leq k\leq m-1$, we deduce that
\begin{equation*}
\begin{aligned}
L_{\overrightarrow{a}}(g^{\frac{k}{m}\cdot(q-1)+(r-1)z})
=& m\cdot \left(g^{\frac{k}{m}\cdot(q-1)+(r-1)z}\right)^{m-1}\cdot \prod\limits_{l\in I, l\neq z}{(g^{(r-1)zm}-g^{(r-1)lm})}.&\\
\end{aligned}
\end{equation*}
Let $u=\prod\limits_{l\in I,l\neq z}(g^{(r-1)zm}-g^{(r-1)lm})$. We can obtain $u=g^{\frac{r+1}{2}\cdot(t-1)-(A+(t-2)z)m+i(r+1)}$ in the same way as
Theorem \ref{thmA1} (i). From $t$ is odd, $A$ and $z$ are even, it follows that $\frac{r+1}{2}\cdot(t-1)-(A+(t-2)z)m+i(r+1)$ is even which implies
$u\in\Box_{q}$. Since $m$ is odd, it implies that $g^{\frac{q-1}{m}}\in\Box_{q}$. Note that $g^{r-1}, m, -1\in \Box_{q}$.
Therefore, $-L_{\overrightarrow{a}}(g^{\frac{k}{m}\cdot(q-1)+(r-1)z})\in\Box_{q}$. Choose $v_{z,k}^{2}=-\left(L_{\overrightarrow{a}}(g^{\frac{k}{m}\cdot(q-1)+(r-1)z})\right)^{-1}$,
with $v_{z,k}\in\mathbb{F}_{q}^{*}$. Define
$$\overrightarrow{v}=(v_{i_{1},0},\ldots,v_{i_{1},m-1},\ldots, v_{i_{t},0},\ldots,v_{i_{t},m-1}).$$
By Lemma \ref{y2}, $\mathbf{GRS}_{\frac{n}{2}}(\overrightarrow{a},\overrightarrow{v},\infty)$ is an MDS self-dual code with length $n=tm+1$.
\end{proof}

\begin{example}
If $r=151$, $q=151^{2}$, $m=15$ and $t=67$, then $\frac{r+1}{2\gcd(r+1,m)}=76>67=t$. By Theorem \ref{thmA2}, there exists an MDS self-dual code of
length $n=tm+1=1006$. This is a new parameter of MDS self-dual code which has not been covered by previous works.
\end{example}

\begin{theorem}\label{thmC1}
Let $q=r^{2}$, where $r$ is an odd prime power. Let $m\mid q-1$, $s$ even, $s\mid m$ and $s\mid r+1$. For $1\leq t\leq \frac{s(r-1)}{\gcd(s(r-1),m)}$,

(i). if $n=tm$, both $\frac{q-1}{m}$ and $\frac{r+1}{s}$ are even, then there exists a $q$-ary $[n,\frac{n}{2}]$ MDS self-dual code.

(ii). if $n=tm+2$, then there exists a $q$-ary $[n,\frac{n}{2}]$ MDS self-dual code.
\end{theorem}
\begin{proof}
(i). In (\ref{ST}), put $S=\langle g^{\frac{r+1}{s}}\rangle$ and $T=\langle g^{\frac{q-1}{m}}\rangle$.
Let $B=\{g^{\frac{r+1}{s}\cdot i_{1}},\cdots,g^{\frac{r+1}{s}\cdot i_{t}}\}$ be a set of coset representatives of $(S\times T)/T$ with $0\leq i_{1}<\cdots<i_{t}<s(r-1)$. Denote by $I=\{{i_{1},\cdots,i_{t}}\}$ and
$$\overrightarrow{a}=\left(g^{\frac{k}{m}\cdot(q-1)+\frac{r+1}{s}\cdot z}\right)_{0\leq k\leq m-1, z\in I}.$$
Obviously, the entries of $\overrightarrow{a}$ are distinct in $\mathbb{F}_{q}^{*}$. We will show that there exists
$\overrightarrow{v}\in\left(\mathbb{F}_{q}^{*}\right)^{n}$ such that $\mathbf{GRS}_{\frac{n}{2}}(\overrightarrow{a},\overrightarrow{v})$
is an MDS self-dual code of length $n=tm$.

Similarly as Theorem \ref{thmA1} (i),
\begin{equation*}
\begin{aligned}
L_{\overrightarrow{a}}(g^{\frac{k}{m}\cdot(q-1)+\frac{r+1}{s}\cdot z})=
&{\prod\limits_{0 \leq j \leq m-1, j\neq k}{(g^{\frac{k}{m}\cdot(q-1)+\frac{r+1}{s}\cdot z}-g^{\frac{j}{m}\cdot(q-1)+\frac{r+1}{s}\cdot z})}}\cdot&\\
&\prod\limits_{l\in I, l\neq z}{\prod\limits_{j=0}^{m-1}{(g^{\frac{k}{m}\cdot(q-1)+\frac{r+1}{s}\cdot z}-g^{\frac{j}{m}\cdot(q-1)+\frac{r+1}{s}\cdot l})}}&\\
=& m\cdot \left(g^{\frac{k}{m}\cdot(q-1)+\frac{r+1}{s}\cdot z}\right)^{m-1}\cdot \prod\limits_{l\in I, l\neq z}{(g^{\frac{r+1}{s}\cdot zm}-g^{\frac{r+1}{s}\cdot lm})}.&\\
\end{aligned}
\end{equation*}
Let $\xi_{s}$ be a primitive $s$-th root of unity of $\F_q^*$. It is obvious that $g^{\frac{r+1}{s}\cdot r}=\xi_{s}\cdot g^{\frac{r+1}{s}}$.
Let $u=\prod\limits_{l\in I, l\neq z}{(g^{\frac{r+1}{s}\cdot zm}-g^{\frac{r+1}{s}\cdot lm})}$. Since $s\mid m$, it follows that $u^{r}=u$, which implies
$u\in\mathbb{F}_{r}^{*}$. If both $\frac{r+1}{s}$ and $\frac{q-1}{m}$ are even, then we obtain
$L_{\overrightarrow{a}}(g^{\frac{k}{m}\cdot(q-1)+\frac{r+1}{s}\cdot z})\in\Box_{q}$.
Choose $v_{z,k}^{2}=\left(L_{\overrightarrow{a}}(g^{\frac{k}{m}\cdot(q-1)+\frac{r+1}{s}\cdot z})\right)^{-1}$ with $v_{z,k}\in\mathbb{F}_{q}^{*}$.
Define $$\overrightarrow{v}=(v_{i_{1},0},\ldots,v_{i_{1},m-1},\ldots, v_{i_{t},0},\ldots,v_{i_{t},m-1}).$$
According to Lemma \ref{y1}, $\mathbf{GRS}_{\frac{n}{2}}(\overrightarrow{a},\overrightarrow{v})$ is an MDS self-dual code with length $n=tm$.\\

(ii). As in (i), we let
$$\overrightarrow{a}=\left(0,g^{\frac{r+1}{s}\cdot i_{1}},g^{\frac{1}{m}\cdot(q-1)+\frac{r+1}{s}\cdot i_{1}},\ldots,g^{\frac{m-1}{m}\cdot(q-1)+\frac{r+1}{s}\cdot i_{1}},\ldots, g^{\frac{r+1}{s}\cdot i_{t}},\ldots,g^{\frac{m-1}{m}\cdot(q-1)+\frac{r+1}{s}\cdot i_{t}}\right).$$
We will find $\overrightarrow{v}\in\left(\mathbb{F}_{q}^{*}\right)^{n}$ such that
$\mathbf{GRS}_{\frac{n}{2}}(\overrightarrow{a},\overrightarrow{v},\infty)$ is an MDS self-dual code of length $n=tm+2$.

For any $0\leq j\leq m-1$ and for any $l\in I=\{i_{1},\ldots,i_{t}\}$, one has

\begin{equation*}
\begin{aligned}
L_{\overrightarrow{a}}(g^{\frac{k}{m}\cdot(q-1)+\frac{r+1}{s}\cdot z})=
&g^{\frac{k}{m}\cdot(q-1)+\frac{r+1}{s}\cdot z}\cdot{\prod\limits_{0 \leq j \leq m-1, j\neq k}{(g^{\frac{k}{m}\cdot(q-1)+\frac{r+1}{s}\cdot z}-g^{\frac{j}{m}\cdot(q-1)+\frac{r+1}{s}\cdot z})}}\cdot&\\
&\prod\limits_{l\in I, l\neq z}{\prod\limits_{j=0}^{m-1}{(g^{\frac{k}{m}\cdot(q-1)+\frac{r+1}{s}\cdot z}-g^{\frac{j}{m}\cdot(q-1)+\frac{r+1}{s}\cdot l})}}&\\
=& m\cdot g^{\frac{r+1}{s}\cdot zm}\cdot \prod\limits_{l\in I, l\neq z}{(g^{\frac{r+1}{s}\cdot zm}-g^{\frac{r+1}{s}\cdot lm})}&\\
\end{aligned}
\end{equation*}
and
 \[
 \begin{array}{rcl}
 L_{\overrightarrow{a}}(0)
 &=&\prod\limits_{l\in I}\prod\limits_{j=0}^{m-1}\left(0-g^{\frac{j}{m}\cdot(q-1)+\frac{r+1}{s}\cdot l}\right)
=(-1)^{(m+1)t}\cdot\prod\limits_{l\in I}g^{\frac{r+1}{s}\cdot lm}=\pm\prod\limits_{l\in I}g^{\frac{r+1}{s}\cdot lm}.\\[2mm]
\end{array}
 \]
From $s\mid m$, it implies $g^{\frac{r+1}{s}\cdot m}\in\mathbb{F}_{r}^{*}$. Therefore,
$L_{\overrightarrow{a}}(g^{\frac{k}{m}\cdot(q-1)+\frac{r+1}{s}\cdot z}),L_{\overrightarrow{a}}(0)\in\mathbb{F}_{r}^{*}\subseteq\Box_{q}$.
Since $q\equiv1\,(\mathrm{mod}\,4)$,
$-L_{\overrightarrow{a}}(g^{\frac{k}{m}\cdot(q-1)+\frac{r+1}{s}\cdot z}),-L_{\overrightarrow{a}}(0)\in\mathbb{F}_{r}^{*}\subseteq\Box_{q}$.
We choose $v_{z,k}^{2}=-\left(L_{\overrightarrow{a}}(g^{\frac{k}{m}\cdot(q-1)+\frac{r+1}{s}\cdot z})\right)^{-1}$ and
$v_{0}^{2}=-\left(L_{\overrightarrow{a}}(0)\right)^{-1}$, with $v_{z,k},v_{0}\in\mathbb{F}_{q}^{*}$. Define
$$\overrightarrow{v}=(v_{0},v_{i_{1},0},\ldots,v_{i_{1},m-1},\ldots, v_{i_{t},0},\ldots,v_{i_{t},m-1}).$$
According to Lemma \ref{y2}, $\mathbf{GRS}_{\frac{n}{2}}(\overrightarrow{a},\overrightarrow{v},\infty)$ is an MDS self-dual code with length $n=tm+2$.
\end{proof}

\begin{example}
If $r=67$, $q=67^2$, $m=12$, $t=31$ and $s=6$, then both $\frac{r+1}{s}$ and $\frac{q-1}{m}$ are even. Note that $\frac{s(r-1)}{\gcd(s(r-1),m)}=33>31=t$.
By Theorem \ref{thmC1}, there exists an MDS self-dual code of length $n=tm=372$. This MDS self-dual code has not been reported in any previous reference.
\end{example}

\begin{theorem}\label{thmD}
Let $q=p^{2s}$, where $p$ is an odd prime and $s$ is a positive integer. There exists a $q$-ary MDS self-dual code of length $p^{2e}+1$,
where $1\leq e\leq s$.

\end{theorem}

\begin{proof}
Denote by $r=p^{s}$. Let $S=\{\alpha_{1},\alpha_{2},\ldots,\alpha_{p^e}\}$ be an $e$-dimensional $\mathbb{F}_{p}$-vector subspace of $\mathbb{F}_{r}$,
with $1\leq e\leq s$. Choose $\beta\in\mathbb{F}_{q}\backslash\mathbb{F}_{r}$, such that $\beta^{r+1}=1$.
Let $a_{k,j}=\alpha_{k}\beta+\alpha_{j}$, $1\leq k,j\leq p^{e}$ and $\overrightarrow{a}=\left(a_{k,j}:1\leq k,j\leq p^{e}\right)$.
A routine calculation shows that
\begin{equation*}
\begin{aligned}
L_{\overrightarrow{a}}(a_{k_{0},j_{0}})=&\prod\limits_{\substack{1\leq k,j\leq p^{e}\\
(k,j)\neq(k_{0},j_{0})}}\left(a_{k_{0},j_{0}}-a_{k,j}\right)&\\
=&\prod\limits_{\substack{1\leq j\leq p^{e}\\ j\neq j_{0}}}\left(\alpha_{k_{0}}\beta+\alpha_{j_{0}}-\alpha_{k_{0}}\beta-\alpha_{j}\right)\cdot
\prod\limits_{\substack{1\leq k\leq p^{e}\\ k\neq k_{0}}}\left(\alpha_{k_{0}}\beta+\alpha_{j_{0}}-\alpha_{k}\beta-\alpha_{j_{0}}\right)\cdot&\\
&\prod\limits_{\substack{1\leq j\leq p^{e}\\ j\neq j_{0}}}\prod\limits_{\substack{1\leq k\leq p^{e}\\ k\neq k_{0}}}
\left(\alpha_{k_{0}}\beta+\alpha_{j_{0}}-\alpha_{k}\beta-\alpha_{j}\right)&\\
=&\prod\limits_{\substack{1\leq j\leq p^{e}\\ j\neq j_{0}}}\left(\alpha_{j_{0}}-\alpha_{j}\right)\cdot
\prod\limits_{\substack{1\leq k\leq p^{e}\\ k\neq k_{0}}}\left(\left(\alpha_{k_{0}}-\alpha_{k}\right)\beta\right)\cdot
\prod\limits_{\substack{1\leq j\leq p^{e}\\ j\neq j_{0}}}\prod\limits_{\substack{1\leq k\leq p^{e}\\ k\neq k_{0}}}
\left(\left(\alpha_{k_{0}}-\alpha_{k}\right)\beta-\left(\alpha_{j_{0}}-\alpha_{j}\right)\right)&\\
=&\beta^{p^{e}-1}\cdot\prod\limits_{\substack{1\leq j\leq p^{e}\\ j\neq j_{0}}}\left(\alpha_{j_{0}}-\alpha_{j}\right)\cdot
\prod\limits_{\substack{1\leq k\leq p^{e}\\ k\neq k_{0}}}\left(\alpha_{k_{0}}-\alpha_{k}\right)\cdot
\prod\limits_{\substack{1\leq j\leq p^{e}\\ j\neq j_{0}}}\prod\limits_{\substack{1\leq k\leq p^{e}\\ k\neq k_{0}}}
\left(\left(\alpha_{k_{0}}-\alpha_{k}\right)\beta-\left(\alpha_{j_{0}}-\alpha_{j}\right)\right).&
\end{aligned}
\end{equation*}
Since $\alpha_{j_{0}},\alpha_{j},\alpha_{k_{0}},\alpha_{k}\in\mathbb{F}_{r}$ and $\beta\in\Box_{q}$, then
\begin{equation}
\begin{aligned}
\beta^{p^{e}-1}\cdot\prod\limits_{\substack{1\leq j\leq p^{e}\\ j\neq j_{0}}}\left(\alpha_{j_{0}}-\alpha_{j}\right)\cdot
\prod\limits_{\substack{1\leq k\leq p^{e}\\ k\neq k_{0}}}\left(\alpha_{k_{0}}-\alpha_{k}\right)\in\Box_{q}.
\end{aligned}
\end{equation}
Let $u=\prod\limits_{\substack{1\leq j\leq p^{e}\\ j\neq j_{0}}}\prod\limits_{\substack{1\leq k\leq p^{e}\\ k\neq k_{0}}}
\left(\left(\alpha_{k_{0}}-\alpha_{k}\right)\beta-\left(\alpha_{j_{0}}-\alpha_{j}\right)\right)$. Note that
\begin{equation*}
\begin{aligned}
u^{r}=&\prod\limits_{\substack{1\leq j\leq p^{e}\\ j\neq j_{0}}}\prod\limits_{\substack{1\leq k\leq p^{e}\\ k\neq k_{0}}}
\left(\left(\alpha_{k_{0}}-\alpha_{k}\right)\beta^{-1}-\left(\alpha_{j_{0}}-\alpha_{j}\right)\right)&\\
=&(-\beta)^{-(p^{e}-1)^{2}}\cdot\prod\limits_{\substack{1\leq j\leq p^{e}\\ j\neq j_{0}}}\prod\limits_{\substack{1\leq k\leq p^{e}\\ k\neq k_{0}}}
\left(\left(\alpha_{j_{0}}-\alpha_{j}\right)\beta-\left(\alpha_{k_{0}}-\alpha_{k}\right)\right)&\\
=&\beta^{-(p^{e}-1)^{2}}\cdot u.&
\end{aligned}
\end{equation*}
This implies $u^{r-1}=\beta^{-(p^{e}-1)^{2}}$. By $\beta^{r+1}=1$ and $p^{e}-1$ is even, we deduce $u^{(r-1)\cdot\frac{r+1}{2}}=1$,
which yields $u\in\Box_{q}$. By (3), it follows that $L_{\overrightarrow{a}}(a_{k_{0},j_{0}})\in\Box_{q}$.

From $q=r^{2}\equiv1\,(\mathrm{mod}\,4)$, one has $-1\in\Box_{q}$, which implies $-L_{\overrightarrow{a}}(a_{i_{0},j_{0}})\in\Box_{q}$.
We choose $v_{k_{0},j_{0}}^{2}=-\frac{1}{L_{\overrightarrow{a}}(a_{k_{0},j_{0}})}$ with $v_{k_{0},j_{0}}\in\mathbb{F}_{q}^{*}$ and
define $\overrightarrow{v}=(v_{k,j}:1\leq k,j\leq p^{e})$. By Lemma \ref{y2}, $\mathbf{GRS}_{\frac{n}{2}}(\overrightarrow{a},\overrightarrow{v},\infty)$
is an MDS self-dual code of length $p^{2e}+1$.
\end{proof}

\begin{example}\label{ex}
Let $p=3$, $s=5$ and $q=p^{2s}=243^{2}$. We can choose $e=3<5=s$. By Theorem \ref{thmD}, there exists an MDS self-dual code of length
$n=p^{2e}+1=3^{6}+1=730>\sqrt{q}$. The length of this MDS self-dual code is different from all the previous results.
\end{example}

\begin{remark}
In the works [\ref{LLL}] and [\ref{Yan}], any MDS self-dual code with the length of the form $n=tm+1$ satisfies one of three following conditions:
\begin{item}
\item{(1)}. $t=\sqrt{q}$ or $m=\sqrt{q}$, see Theorem 2 (ii), Theorem 3 (i) and (iii) in [\ref{Yan}];
\item{(2)}. $t\mid q-1$ or $m\mid q-1$, see Theorem 2 in [\ref{LLL}];
\item{(3)}. $tm=p^{c}$, $q=p^{k}$ and $c\mid k$, see Theorem 4 (i) in [\ref{Yan}].
\end{item}\\
The class of codes in Theorem \ref{thmD} is not included in the three cases. So it can produce new MDS self-dual codes.
\end{remark}

 \begin{theorem}\label{thmE}
 Let $q=p^{km}$ with $p$ odd prime. For any $t$ with $2t\mid(p^{k}-1)$ and $e\leq m-1$, if $\frac{q-1}{2t}$ is even,
there exists a $q$-ary MDS self-dual code with length $2tp^{ke}$.
 \end{theorem}
 \begin{proof}
 Let $V$ be an $e$-dimensional $\mathbb{F}_{p^k}$-vector subspace in $\mathbb{F}_{q}$ with $V\cap\mathbb{F}_{p^k}=0$.
Let $\omega\in\mathbb{F}_{p^k}$ be a primitive element of order $2t$. Choose $\overrightarrow{a}=\bigcup\limits_{j=0}^{2t-1}(\omega^{j}+V)$.
For any $b\in\omega^{i}+V$,
\begin{equation*}
\begin{aligned}
L_{\overrightarrow{a}}(b)&=\left(\prod_{0\neq u\in V}u\right)\cdot\left(\prod_{j=0,j\neq i}^{2t-1}\prod_{u\in V}(\omega^{i}-\omega^{j}+u)\right)&\\
&=\left(\prod_{0\neq u\in V}u\right)\cdot\left(\prod_{u\in V}\omega^{i(2t-1)}\prod_{h=1}^{2t-1}\left(1+\omega^{-i}u-\omega^{h}\right)\right)&\\
&=\omega^{-ip^{ke}}\cdot\left(\prod_{0\neq u\in V}u\right)\cdot\left(\prod_{u\in V}\prod_{h=1}^{2t-1}(1+u-\omega^{h})\right)&
\end{aligned}
\end{equation*}
where the last equality follows from that $\prod\limits_{u\in V}\omega^{i(2t-1)}=\omega^{-ip^{ke}}$ and $\omega^{-i}u$ runs through $V$
when $u$ runs through $V$.

Let $c=\left(\prod\limits_{0\neq u\in V}u\right)\cdot\left(\prod\limits_{u\in V}\prod\limits_{h=1}^{2t-1}(1+u-\omega^{h})\right)$.
It follows that $L_{\overrightarrow{a}}(b)=\omega^{-ip^{ke}}\cdot c$. Note that $\omega\in\Box_{q}$, since $\frac{q-1}{2t}$ is even. We can choose
$\lambda=c$, which is independent of $b$. Let $v_{b}^{2}=(\lambda L_{\overrightarrow{a}}(b))^{-1}$, with $v_{b}\in\mathbb{F}_{q}^{*}$ and
define $\overrightarrow{v}=(v_{b}:b\in\omega^{i}+V)$. By Lemma \ref{y1}, $\mathbf{GRS}_{\frac{n}{2}}(\overrightarrow{a},\overrightarrow{v})$
is an MDS self-dual code with length $2tp^{ke}$.
 \end{proof}

\begin{example}
Let $p=5$, $k=3$, $m=9$ and $q=p^{km}=5^{27}$. We can choose $t=31$ and $e=7$. It is easy to verify that $2t\mid p^{k}-1$, $e\leq(m-1)k$ and
$\frac{q-1}{2t}$ is even. By Theorem \ref{thmE}, there exists an MDS self-dual code of length $n=2tp^{e}=62\times5^{21}$.
This code has not been reported in any previous work.
\end{example}

\begin{example}
For $q=151^{2}$, we can construct $862$ different $n$ for which MDS self-dual codes of length $n$ by using all the previous results (in Table 1).
Utilizing the results in this paper (Theorems \ref{thmA1}-\ref{thmE}), we can construct $1228$ MDS self-dual codes of different
lengths. Usually, for large $q$ being square of odd prime power, we can produce much more MDS self-dual codes over $\mathbb{F}_q$ than the total of previous results.
\end{example}

\section{Conclusion}

\quad\; Based on the technique in [\ref{JX}], [\ref{LLL}] and [\ref{Yan}] and applying the second fundamental theorem of group homomorphism on
different multiplicative subgroups of $\mathbb{F}_{q}^{*}$, we construct several new classes of MDS self-dual codes over finite fields of
odd characteristic via generalized Reed-Solomon codes and extended generalized Reed-Solomon codes. For a fixed odd prime power $q$ and any even
$n\leq q+1$, utilizing $\mathbf{GRS}$ codes and extended $\mathbf{GRS}$ codes, we hope to construct MDS self-dual code with length $n$. So the number of
$q$-ary MDS self dual codes with different lengths is expected to be $\frac{q+1}{2}$ except that $q\equiv 3\pmod{4}$ and $n\equiv 2\pmod{4}$ (in this case, there does not exist MDS self-dual codes, see [\ref{ZF}]). However, the total number of MDS self-dual codes in all known
results is much less than $\frac{q+1}{2}$. Therefore, much more MDS self-dual codes over finite fields of odd characteristic are yet to be explored.

\section*{Acknowledgements}
{The authors thank the editor and anonymous referees for their work to improve the readability of this paper.  This work is partially supported by National Natural Science Foundation of China(NSFC) under Grant 11471008(J.Luo) and Grant 11871025(H.Liu) and also supported by the self-determined research funds of
CCNU from the self-determined research funds of CCNU from the colleges' basic research and operation of MOE(Grant No. CCNU18TS028).
}

\end{document}